\documentclass{article}
\usepackage{authblk}
\usepackage{verbatim}
\usepackage{graphicx, color} % Required for inserting images
\usepackage{xcolor}
\usepackage{enumerate}
\usepackage{amsmath}
\usepackage{calligra}
\usepackage{stackengine}
\usepackage{amssymb}
\usepackage{geometry}
\usepackage{float}
\usepackage{booktabs}
\usepackage{multirow}
\usepackage{multicol}
\usepackage{pdflscape}
\usepackage{bm}
\usepackage{commath}
\usepackage{mathrsfs}
\usepackage{amsthm}
\usepackage[round, sort&compress]{natbib}
\usepackage{natbib}
\setcitestyle{authoryear,open={(},close={)}}

\usepackage[noend]{algpseudocode}[1]
\usepackage{algorithmicx}
\usepackage{algorithm}
\allowdisplaybreaks

\newcommand*\Let[2]{\State #1 $\gets$ #2}
\newcommand*\LetIn[2]{ #1 $\gets$ #2}

\algrenewcommand\algorithmicrequire{\textbf{Input:}}
\algrenewcommand\algorithmicensure{\textbf{Output:}}

\makeatletter
\renewcommand{\Procedure}[2]{%
  \csname ALG@cmd@\ALG@L @Procedure\endcsname{#1}{#2}%
  \def\jayden@currentfunction{#1}%
}
\renewcommand{\Function}[2]{%
  \csname ALG@cmd@\ALG@L @Function\endcsname{#1}{#2}%
  \def\jayden@currentfunction{#1}%
}
\newcommand{\funclabel}[1]{%
  \@bsphack
  \protected@write\@auxout{}{%
    \string\newlabel{#1}{{\textsc{\jayden@currentfunction}}{\thepage}}%
  }%
  \@esphack
}
\makeatother

\usepackage[textsize=tiny]{todonotes}

\newtheorem{theorem}{Theorem}[section] % Numbered within sections
\DeclareMathOperator{\conv}{conv}
\newcommand{\R}{\mathbb{R}}
\newcommand{\B}{\mathbb{B}}
\newcommand{\Z}{\mathbb{Z}}
\newcommand{\NP}{\mathsf{NP}}
\newcommand{\Pclass}{\mathsf{P}}
\newcommand{\APX}{\mathsf{APX}}
\newcommand{\bigO}{\mathcal{O}}
\DeclareMathOperator{\poly}{poly}
\DeclareMathOperator{\OPT}{OPT}

\newtheorem{lemma}[theorem]{Lemma}     % Shares numbering with theorems
\newtheorem{proposition}[theorem]{Proposition}
\newtheorem{corollary}[theorem]{Corollary}

\theoremstyle{definition}
\newtheorem{definition}[theorem]{Definition}

\theoremstyle{remark}

\allowdisplaybreaks

\newcommand{\mkc}{\textsc{mkc}}
\newcommand{\mkp}{\textsc{mkp}}
\newcommand{\mkcone}{\textsc{1-mkc}}

\theoremstyle{theorem}
\newtheorem*{mkc-problem}{Multidimensional knapsack cover problem with one continuous variable per dimension \rm{(\mkc)}}{\bf}{\rm}
\newcommand{\p}{\textsc{p}}
\newcommand{\pone}{\textsc{p}\textsubscript{1}}
\newcommand{\pu}{\textsc{p}\textsubscript{u}}
\newcommand{\q}{\textsc{q}}

\newcommand{\cip}{\textsc{cip}}
\newcommand{\pip}{\textsc{pip}}

\newcommand{\gap}{\textsc{gap}}

\usepackage[noend]{algpseudocode} 
\usepackage{microtype}

\title{Covering and packing mixed-integer linear programs with a fixed number of constraints: Approximation and convex hull}
\author{Kobe Grobben}
\author{Phablo F.~S. Moura}
\author{Hande Yaman}

\affil{Research Center for Operations Research \& Statistics, KU Leuven,
Belgium}

\begin{document}

\maketitle

\begin{abstract} 
This paper presents an algorithmic study of a class of covering mixed-integer linear programming problems which encompasses classic cover problems, including multidimensional knapsack, facility location and supplier selection problems. We first show some properties of  optimal solutions, which are then used to decompose the problem into instances of the multidimensional knapsack cover problem with a single continuous variable per dimension. The proposed decomposition is used to design a polynomial-time approximation scheme for the problem with a fixed number of constraints. To the best of our knowledge, this is the first approximation scheme for such a general class of covering mixed-integer linear programs.
Moreover, we design a fully polynomial-time approximation scheme and an approximate linear programming formulation for the case with a single constraint.
These results improve upon the previously best-known 2-approximation algorithm for the knapsack cover problem with a single continuous variable.
Finally, we show a perfect compact formulation for the case where all variables have the same lower and upper bounds.
Analogous results are derived for the packing and more general variants of the problem.
\end{abstract}
\section{Introduction}
This paper presents an algorithmic study of mixed-integer linear programming problems, denoted by~\p, which are defined as:
\begin{align}
        \min  \;\;  &  \sum_{i\in [n]} \sum_{j \in [m]} v_{ij} x_{ij} + \sum_{i \in [n]} f_i y_i \label{objective function} &\\
    \text{s.t.}\;\; &  \sum_{i \in [n]} x_{ij} \ge d_j  & \forall j \in [m], \label{demand}\\
                & \ell_{ij} y_i \le x_{ij} \le c_{ij} y_i & \forall i \in [n], \; j \in [m], \label{Lower&upper}\\
                & y_i \in \{0,1\} & \forall i \in [n], \label{Integer}
\end{align}
where $m$ and $n$ are positive integers, $v, \ell, c \in \mathbb{Z}_\ge^{nm}$, $d \in \mathbb{Z}_>^{m}$ and $f \in \mathbb{Z}_\ge^{n}$ denote the nonnegative integer parameters with $c_{ij} \geq \ell_{ij}$  for all $i \in [n]$ and $j \in [m]$ and $[a]=\{1, \ldots, a\}$ for any positive integer $a$. 
This model is applicable to a range of problems, including 
variants of the knapsack problem and the supplier selection problem. In the latter, the aim is to cover the demand for $m$ items, where item $j\in [m]$ has demand $d_j$,  by choosing a subset of suppliers from a given set of $n$ candidates, subject to minimum and maximum order quantities $\ell_{ij}$ and $c_{ij}$ for each supplier $i\in [n]$ and each item $j\in [m]$, and to minimize the total cost, where $f_i$ is the fixed cost of selecting supplier $i\in [n]$ and $v_{ij}$ is the unit cost of an item $j\in [m]$ ordered from supplier $i\in [n]$. 

For completeness, we also consider the packing variant of~\p, defined analogously as: 
\begin{align*}
        \max  \;\;  &  \sum_{i\in [n]} \sum_{j \in [m]} v_{ij} x_{ij} + \sum_{i \in [n]} f_i y_i &\\
    \text{s.t.}\;\; &  \sum_{i \in [n]} x_{ij} \le d_j  & \forall j \in [m],\\
                & \ell_{ij} y_i \le x_{ij} \le c_{ij} y_i & \forall i \in [n], \; j \in [m],\\
                & y_i \in \{0,1\} & \forall i \in [n],
\end{align*}
under the same assumptions on all parameters.

The covering integer programs (\cip) are closely related to \p.
Given a matrix~$A \in \R^{m\times n}_\ge$, vectors $a \in \R^m_\ge$, and $b, h \in \R^n_\ge$, \cip\ is defined as $\min\{h^T z : Az \ge a, z \le b, z \in \Z^{n}_\ge\}$.
Carr et al.~\cite{carr2000strengthening} propose a $\Delta_1$-approximation for this problem, where $\Delta_1$ is the maximum number of non-zero coefficients in any row of~$A$.
Kolliopoulos and Young~\cite{kolliopoulos2005approximation} design a $\bigO(\log m)$-approximation for \cip, and observe that this is asymptotically the best possible unless $\Pclass=\NP$.
This inapproximability threshold follows from the classic set cover problem, which does not admit a $o(\log m)$-approximation unless $\Pclass=\NP$~\citep{raz1997sub}.
The set cover problem is a special case of~\cip\ where~$z$ is binary, and binary \cip\ is a particular case of~\p\ where $\ell_{ij}=c_{ij}=A_{ij}$, $v_{ij}=0$, and $f_i=h_i$ for all $i \in [n]$ and $j\in [m]$.
As a consequence, \p\ does not admit a $o(\log m)$-approximation  unless $\Pclass=\NP$. Furthermore, for arbitrary $m > 1$, no constant approximation ratio can be obtained for~\p\ unless $\Pclass = \NP$ \citep{srinivasan1995improved}. 

Although there is little hope of designing better approximations for~\p\ in general, the case where the number of constraints is fixed remains largely unexplored.
Observe that~\p\ with fixed~$m$ encompasses classic cover problems, including multidimensional knapsack cover problems and therefore, \p\ remains $\NP$-hard even when~$m=1$.
For any fixed $m\geq 2$,~\p\ with~$m$ constraints does not admit a fully polynomial-time approximation scheme (FPTAS)  unless $\Pclass = \NP$.
This is due to the analogous hardness result for multidimensional knapsack cover that can be obtained following essentially the same proof for the multidimensional knapsack (packing) devised by Magazine and Chern~\cite{magazine1984note}, and to the fact that multidimensional knapsack cover is a particular case of \p\ with a fixed number of constraints.
On the positive side, Frieze and Clarke~\cite{friezeApproximationAlgorithmsMdimensional1984} design a  polynomial-time approximation scheme (PTAS) to the multidimensional knapsack cover problem. Kulik and Shachnai~\cite{kulikThereNoEPTAS2010} show that no efficient PTAS (EPTAS) exists for this problem unless $W[1] =$ FPT. 
For the knapsack cover problem (i.e., the one-dimensional case), Güntzer and Jungnickel~\cite{guntzer2000approximate} propose an FPTAS which is based on a greedy strategy. 

The packing integer programs (\pip) are defined analogously and encompass classic combinatorial optimization problems such as knapsack, matching, and independent set problems.
As a consequence, binary \pip\ (and thus the packing variant of \p) cannot be approximated within $n^{1-\epsilon}$, for any $\epsilon >0$, unless $\Pclass=\NP$~\citep{zuckerman2006linear}.
There exists a vast literature on approximation algorithms for binary \pip, nearly all of which rely on randomized rounding techniques~\citep[see, e.g.,][]{ bansal2012solving,chekuri2004multidimensional, srinivasan1995improved}.

We also investigate a more general class of packing mixed-integer programs, defined as follows:
\begin{align*}
        \max  \;\;  &  \sum_{i\in [n]} \sum_{j \in [m]} v_{ij} x_{ij} + \sum_{i \in [n]} \sum_{j \in [m]}f_{ij} y_{ij} & \notag\\
    \text{s.t.}\;\; &  \sum_{i \in [n]} x_{ij} \le d_j  & \forall j \in [m], \\ 
    &  \sum_{j \in [m]} y_{ij} \le t_i  & \forall i \in [n], \\
                & \ell_{ij} y_{ij} \le x_{ij} \le c_{ij} y_{ij} & \forall i \in [n], \; j \in [m], \notag\\
                & y_{ij} \in \{0,1\} & \forall i \in [n], \; j\in [m], \notag
\end{align*}
where $t \in \Z^n_>$ and the other parameters are defined as before. 
Note that this includes relevant problems as the Generalized Assignment Problem (\gap). In particular, when $t_i = 1$, $v_{ij} = 0$ and $c_{ij} = \ell_{ij}$ for all $i \in [n]$ and $j \in [m]$, this problem is equivalent to \gap, which is known to be $\APX$-hard~\citep{chekuri2005polynomial}. 
On the positive side,  for any $\epsilon>0$, there is an approximation for \gap\ with ratio $(e/(e-1)+ \epsilon) \approx (1.582 + \epsilon)$  using an LP-rounding approach, which is currently the best known approximation for this problem~\citep{fleischerTightApproximationAlgorithms2006}.

\subsection*{Contributions}
This work provides a theoretical study of the three problem classes introduced above through the lens of approximation algorithms and perfect linear programming formulations.
In Section~\ref{sec:decomposition}, we first show some properties of some  optimal solutions of~\p.
Using these properties, we decompose \p\ into instances of a generalization of the multidimensional knapsack cover problem with one continuous variable per dimension.
The proposed decomposition of \p\ is used to design a PTAS for \p\ with a fixed number of constraints in Section~\ref{sec:ptas}.
To the best of our knowledge, this is the first approximation scheme for such a general class of covering mixed-integer linear programs. 
Furthermore, in a sense, it is the best possible approximation for this problem since \p\ admits no FPTAS, as previously observed. In Section~\ref{sec:PTAS_Packing},
using a similar strategy, we also design PTASes for the packing variants of \p.

For \p\ with a single constraint, we devise an FPTAS based on a dynamic programming approach in Section~\ref{sec:fptas}.
A byproduct of this result is an  FPTAS for the knapsack cover with a single continuous variable, thereby improving upon the previously best-known 2-approximation algorithm by Zhao and Li~\cite{zhaoApproximationAlgorithms012014}.
We note that Van Hoesel and Wagelmans~\cite{vanhoeselFullyPolynomialApproximation2001} present an FPTAS for the single-item capacitated economic lot-sizing problem that is more general than one-dimensional~\p\ in its consideration of backlogging and holding costs.
However, their model does not consider lower bounds.
Additionally, for each $\epsilon  \in (0,1)$, we prove the existence of a compact linear programming formulation that yields a solution of value at most $(1+\epsilon)$ times the optimal value of the problem. 

Finally, in Section~\ref{Section:extended}, we provide a perfect compact formulation for the one-dimensional \p\ with uniform bounds, that is, $\ell_{i} = \ell$ and $c_{i} = c $ for all $i \in [n]$.
Constantino~\cite{constantinoLowerBoundsLotSizing1998a} investigates the uncapacitated version of this problem, i.e., the case where $c\geq d$. 
More precisely, he derives two families of valid inequalities and shows that they describe the convex hull in this special case.
Although one-dimensional uniform \p\ is a special  case (with positive demand only in the last period) of a  single-item lot-sizing problem with piecewise linear costs, which is known to be polynomial-time solvable~\citep{hellionPolynomialTimeAlgorithm2012a}, no tight formulation is known for this problem. 
Hence the proposed perfect compact formulation contributes in this direction.
Concluding remarks and directions for further investigation are presented in Section~\ref{sec:conclusion}. Note that each section also includes an analysis of the corresponding packing problem.

\section{Decomposition into multidimensional knapsack problems}\label{sec:decomposition}
This section begins by showing a simple characterization of the instances of \p\ with optimal value equal to zero.

\begin{proposition}
    Let $I$ be an instance of \p, and let $S=\{i \in [n] : f_i=0\}$.
    It holds that the optimal value $\OPT(I)=0$ if and only if, for each $j \in [m]$, $\sum_{i \in Z_j} c_{ij} \ge d_j$,
    where $Z_j=\{i \in S :  v_{ij}=0\}$.
\end{proposition}
\begin{proof}
    Suppose first that $\OPT(I)=0$, and let $(x,y)$ be an optimal solution of $I$.
    Clearly, $f_i=0$ for each~$i \in [n]$ with $y_i=1$.
    Consider $i \in [n]$ with $y_i=1$, and $j \in [m]$.
    If $v_{ij}>0$, then $x_{ij}=0$.
    Hence, $\sum_{i \in Z_{j}} c_{ij} \ge \sum_{i \in Z_j} x_{ij}=\sum_{i \in [n]} x_{ij} \geq d_j$. 

    For the converse, define a vector $(x,y) \in \R^{nm}\times \{0,1\}^n$ that equals zero except for the entries $y_i=1$ for all $i \in [n]$ with $f_i=0$, and $x_{ij}=c_{ij}$ for all $j \in [m]$ and $i  \in Z_{j}$.
    If follows from the construction that $(x,y)$ is a feasible solution of $I$ of value equal to zero.
    Hence, $\OPT(I)=0$.

\end{proof}
\noindent Using the previous result, one may recognize and solve any instance with optimal value equal to zero in~$\bigO(nm)$ time.
Henceforth, we assume that all considered instances of \p\ have a positive optimal value.

We next present properties of an optimal solution of~\p, which allow us to decompose~\p~into instances of the following generalization of the multidimensional knapsack cover problem.
In this problem, we have $\eta$ items to choose from to cover $\mu$ dimensions, where item $i\in [\eta]$ has weight $\bar w_{ij}$ in dimension $j\in [\mu]$, the cost of choosing item $i\in [\eta]$ is $\bar f_i$ and the demand to be covered in dimension $j\in [\mu]$ is $\bar d_j$. In addition, we can use at most $\bar c_j $ amount of a resource to cover the demand in dimension $j\in [\mu]$ at a cost of $\bar v_j$ per unit. The aim is to cover the demand in all dimensions at minimum cost. More precisely,  

\noindent\begin{mkc-problem}\hfill\\
{\rm\textsc{Input:}} $\eta, \mu \in \Z_>$,  $\bar f \in \Z^\eta_\ge$, $\bar v \in \Z^\mu_\ge$, $\bar c \in \Z^\mu_>$, $\bar w \in \Z^{\eta\mu}_\geq$, and $\bar d \in \Z^\mu_\ge$.\\
{\rm\textsc{Output:}} Subset $S \subseteq [\eta]$,  and nonnegative $\alpha_{j} \leq \bar c_{j}$ for all $j \in [\mu]$ such that $\sum_{ i \in S} \bar w_{ij} \ge \bar d_j -\alpha_{j}$ for every~$j \in [\mu]$.\\ 
{\rm\textsc{Objective:}} Minimize $\sum_{i \in S} \bar f_i + \sum_{j \in [\mu]}\bar v_{j} \alpha_{j}$.
\end{mkc-problem}
\noindent Note that an instance of \mkc\ is feasible if and only if $\sum_{ i \in [\eta]} \bar w_{ij} \ge \bar d_j -\bar c_{j}$ for every $j \in [\mu]$.
We assume henceforth that all considered instances of \mkc\ are feasible.

Let $X=\{(x,y) \in \R^{nm}\times \{0,1\}^n : (x,y) \text{ satisfies }\eqref{demand} - \eqref{Integer} \}$. Let $j \in [m]$ and $k \in [n]$.
We define 
$L_{kj} = \{ i \in [n] \setminus \{k\} : v_{ij} > v_{kj} \} \cup \{ i \in [n] \setminus \{k\} :   v_{ij} = v_{kj} \text{ and } i < k\}$ and $C_{kj} = \{ i \in [n] \setminus \{k\} : v_{ij} < v_{kj}\} \cup \{ i \in [n] \setminus \{k\} :  v_{ij} = v_{kj} \text{ and } i > k\}.$
Note that $\{L_{kj}, C_{kj}\}$ is a partition of $[n]\setminus\{k\}$ which only depends on the input instance, and can be easily computed in linear time. The key theorem is presented next.

\begin{theorem}\label{thm:nice-optimal}
    Let $I$ be an instance of~\p.
    There exists a function $g \colon [m]\to[n]$ such that $I$ admits an optimal solution $(x,y) \in \R^{nm} \times \{0,1\}^n$ where, for each $j \in [m]$, 
    $y_{g(j)}=1$, $x_{ij}=\ell_{ij}y_i$ for all $i\in L_{g(j),j}$ and~$x_{ij}=c_{ij}y_i$ for all $i\in C_{g(j),j}$.
\end{theorem}
\begin{proof}
Let $(x,y) \in \R^{nm + n}$ be an optimal solution to~\p\ that is  an extreme point of $conv(X)$, and let $j \in [m]$.
Define $N_1 = \{ i \in [n] : y_i = 1, x_{ij}=\ell_{ij}\}$, $N_2 = \{ i \in [n] : y_i = 1, x_{ij}=c_{ij}\}$, and $N_3 = \{ i \in [n] : \ell_{ij} < x_{ij} < c_{ij}\}$.
Note that we cannot have $N_1=N_2=N_3=\emptyset$ as $d >0$.

Suppose to the contrary that $\ell_{a j} < x_{a j} < c_{a j}$ and $\ell_{b j} < x_{b j} < c_{b j}$ for some %$j\in [m]$ and
    $a, b \in [n]$ with $a\neq b$. 
    Let $\epsilon = \min_{i \in \{a,b\}}\min\{c_{ij}-x_{ij}, \: x_{ij} - \ell_{ij}\}$.
    We define two vectors $x^1, x^2 \in \R^{nm+n}$ such that $x^1$ and $x^2$ are equal to $x$ except for the following entries:
    $x_{a j}^{1} = x_{a j} + \epsilon$, $x_{b j}^{1} = x_{b j} - \epsilon$, 
    $x_{a j}^{2} = x_{a j} - \epsilon$,   and $x_{b j}^{2} = x_{b j} + \epsilon$.
It is clear that $(x^1,y)$ and $(x^2,y)$ belong to $X$.
Moreover, $(x,y)=\frac{1}{2} (x^1,y) + \frac{1}{2} (x^2,y)$, a contradiction to the choice of $(x,y)$. Hence, there exists at most one $a\in [n]$
 with $\ell_{a j} < x_{aj} < c_{aj}$, and so $|N_3|\leq 1$. 

We define $g(j)=i$ if $N_3=\{i\}$.
Otherwise, it holds that $N_3=\emptyset$, and we define
\[ g(j) = \begin{cases} \displaystyle \arg\max_{i \in N_2} \{v_{ij}\}  & \parbox[t]{0.55\linewidth}{\raggedright if $N_2\neq \emptyset$,} \\[10pt] 
\displaystyle \arg\min_{i \in N_1}  \{ v_{ij}\} & \text{otherwise}. \end{cases} \]
We next prove that the following two assertions hold.
    \begin{enumerate}[(i)]
        \item For each $i \in [n]\setminus\{g(j)\}$  such that $y_i=1$ and $v_{ij}\neq v_{g(j),j}$, we have $x_{ij}=\ell_{ij}$ if $v_{ij} > v_{g(j),j}$  
         and $x_{ij}=c_{ij}$ 
        otherwise.
        \label{it:trivial-assignment}
        \item For each $i \in [n]\setminus\{g(j)\}$  such that $y_i=1$ and $v_{ij}=v_{g(j),j}$, we have $x_{ij}=\ell_{ij}$ if~$i < g(j)$  
        and $x_{ij}=c_{ij}$ otherwise.
        \label{it:coherence}
    \end{enumerate}

Suppose first that $\ell_{g(j),j} < x_{g(j),j} < c_{g(j),j}$.
Suppose to the contrary that $y_i=1$, $x_{ij}=c_{ij} > \ell_{ij}$  and $v_{ij}>v_{g(j),j}$ for some $i \in [n]\setminus\{g(j)\}$.
Define a vector $\bar  x \in \R^{nm}$ to be  equal to $x$ except for $\bar x_{ij} = x_{ij}-\epsilon$ and $\bar x_{g(j),j} = x_{g(j),j} +\epsilon$, where $\epsilon:=\min\{c_{g(j),j} - x_{g(j),j}, \:  x_{ij}-\ell_{ij}\}$.
It follows that $(\bar x,y)$ is a feasible solution to~\p~of cost strictly smaller than $(x,y)$, a contradiction.
Analogously, one can prove that $x_{ij}=\ell_{ij}$ if $y_i=1$ and  $v_{ij}>v_{g(j),j}$.
Suppose now that $x_{ij} \in \{\ell_{ij}y_i, c_{ij}y_i\}$ for all $i \in [n]$.
Note that $v_{ij} \ge v_{i'j}$ for all $i, i' \in [n]$ such that $y_i=y_{i'}=1$, $x_{ij}=\ell_{ij}<c_{ij}$ and $x_{i'j}= c_{i'j}>\ell_{i'j}$.
Otherwise, one could obtain a solution cheaper than $(x,y)$ by increasing~$x_{ij}$ and decreasing~$x_{i'j}$ by the same (sufficiently small) amount.
Let $i \in [n]\setminus\{g(j)\}$ such that $y_i=1$ and $v_{ij}\neq v_{g(j),j}$.
If $x_{g(j),j} = c_{g(j),j}$, then the previous observation together with the definition of $g(j)$ imply that $x_{ij}=c_{ij}$ if $v_{ij}<v_{g(j),j}$, and $x_{ij}=\ell_{ij}$ if $v_{ij}>v_{g(j),j}$.
If $x_{g(j),j} = \ell_{g(j),j}$, then there is no $i \in [n]$ with $y_i=1$ and $x_{ij}=c_{ij}$. Thus~\eqref{it:trivial-assignment} follows from the choice of $g(j)$.

We define the \emph{set of inversions} of any solution $(x',y')$ as
\[\{i \in \{1, \ldots, g(j)-1\} : x'_{ij}=c_{ij}, y'_i=1\} \cup \{i \in \{g(j)+1, \ldots, n\} : x'_{ij}=\ell_{ij}, y'_i=1\}.\]
Let $V$ be the set of inversion of $(x,y)$.
Suppose that $(x,y)$ has at least one inversion (i.e., $V\neq \emptyset$).
Assume, without loss of generality, that this inversion is due to an element in $[n]$ smaller than $g(j)$, and choose~$i \in V$ to be the largest number such that $i < g(j)$.
For each $i' \in \{i+1, \ldots, g(j)\}$ with $y_{i'}=1$ and $v_{i'j}=v_{g(j),j}$, starting from $g(j)$ to $i+1$, we first move from $x_{ij}$ to $x_{i'j}$ the maximum quantity~$\epsilon$ such that either $x_{ij}-\epsilon=\ell_{ij}$ or $x_{i'j}+\epsilon = c_{i'j}$. Then we are done in the former case, or we repeat the procedure with~$i'-1$ or end the procedure when~$i'-1 = i$ in the latter one.
Let $\bar x \in \R^{nm}$ be the vector obtained at the end of this procedure.
Note that there is $t \in \{i, i+1, \ldots, g(j)\}$ with $y_{t}=1$ and $v_{tj}=v_{g(j),j}$ such that, for each $i' \in \{i, \ldots, g(j)\}$ with $y_{i'}=1$ and $v_{i'j}=v_{g(j),j}$, we have  $\bar x_{i' j} = \ell_{i'j}$ if~$i' < t$ and $\bar x_{i' j} = c_{i'j}$ if~$i' > t$.
Hence, the number of inversions of $(\bar x, y)$ is strictly smaller than the number of inversions of $(x,y)$.
Moreover, since the procedure only changes entries $i \in [n]$ such that $v_{ij}=v_{g(j),j}$, it is clear that~$(\bar x, y)$ is an optimal solution to~\p~that still respects~\eqref{it:trivial-assignment}.
By repeating the previous procedure at most~$n-1$ times, we obtain an optimal solution to~\p~which satisfies~\eqref{it:trivial-assignment}, and~\eqref{it:coherence}. The case of an inversion involving an element in $[n]$ larger than $g(j)$ is handled analogously.

\end{proof}

\noindent It follows from Theorem~\ref{thm:nice-optimal} that problem~\p\ boils down to finding such a function~$g$ and solving the \mkc\ instance associated with~$g$.
Using this idea, we next propose an alternative formulation for~\p~applied to  the (sub)partitions of $[n]$ given by $L$ and $C$ for all possible $\bigO(n^m)$ choices of $(g(1), \ldots, g(m)) \in [n]^{m}$.

Let $\mathcal{G}$ be the set of all functions $g \colon [m]\to [n]$.
For every $g\in \mathcal{G}$, $i \in [n]$ and $j \in [m]$, define
\[w_{ij}^g = 
\begin{cases}
\ell_{ij} & \text{ if $i \in L_{g(j), j}$,}\\
c_{ij} & \text{ if $i \in C_{g(j), j}$,} \\
\ell_{g(j)j} & \text{ if $i= g(j)$.} 
\end{cases}\]
Let $g \in \mathcal{G}$.
For each $j \in [m]$,  define $c^g_{g(j),j} = c_{g(j),j} - \ell_{g(j),j}$, and $v^g_j = v_{g(j),j}$.
For each $i \in [n]$,  define $f^g_i = f_i + \sum_{j\in [m]} v_{ij} w^g_{ij}$.

The following mixed-integer linear formulation has binary variables $z^g$ and $y^g_i$ for all $g \in \mathcal{G}$ and $i \in [n]$, and real variables $\alpha^g_j$ for all $g \in \mathcal{G}$ and $j \in [m]$.

\begin{align}
    \min  \;\;  & \sum_{g \in \mathcal{G}} \left(\sum_{i \in [n]}  f^g_i y_i^g + \sum_{j \in [m]} v_j^g \alpha_j^g \right)  & \label{Multi1}\\
    \text{s.t.}\;\;             & \sum_{g \in \mathcal{G} }z^g = 1, &  \label{Multi2}\\
    &y^g_i \le z^g & \forall g \in \mathcal{G}, i \in [n],\\
    & y^g_{g(j)} = z^g & \forall g \in \mathcal{G}, j \in [m], \\
     & \sum_{i\in [n]} w^g_{ij} y^g_i + \alpha^g_j \ge d_j z^g   & \forall g \in \mathcal{G}, j \in [m],  \\
                &  c^g_{g(j), j} z^g \ge \alpha^g_j \ge 0 & \forall g \in \mathcal{G}, j \in [m],\\
                % & \sum_{g \in G }z^g = 1, & \\
                & z^g \in \{0,1\} & \forall g \in \mathcal{G}, \\
                & y^g_i \in \{0,1\} & \forall g \in \mathcal{G}, i \in [n]. \label{Multi9}
\end{align}

\begin{proposition}\label{prop:union-polyhedral-model}
    The formulation \eqref{Multi1}--\eqref{Multi9} correctly models~\p. 
\end{proposition}
\begin{proof}
    Let~$I$ be an instance of \p.
    Let $g\colon [m]\to [n]$ and $( x, y)$ as given in Theorem~\ref{thm:nice-optimal}.
    Let us define a vector $(z,\bar y,\alpha) \in \{0,1\}^{|\mathcal{G}|}\times \{0,1\}^{|\mathcal{G}|n} \times \R^{|\mathcal{G}|m}$ with all entries equal to zero except for $z^g=1$, $\bar y^g= y$, and $\alpha^g_j =  x_{g(j),j} - \ell_{g(j),j}$ for all $j \in [m]$.
    One may easily verify that $(z,\bar y,\alpha)$ satisfies constraints~\eqref{Multi2}--\eqref{Multi9} as $(x, y)$ satisfies~\eqref{demand}--\eqref{Integer}.
    Thus, the objective function~\eqref{Multi1} on $(z,\bar y,\alpha)$ is equal to
    \begin{align}
    \sum_{i \in [n]}  f^g_i \bar y_i^g + \sum_{j \in [m]} v_j^g \alpha_j^g & = \sum_{i\in [n]} \sum_{j \in [m]} v_{ij} w^g_{ij} \bar y^g_i + \sum_{i\in [n]} f_i \bar y^g_i + \sum_{j \in [m]} v^g_{j}\alpha^g_{j}\nonumber\\
    \begin{split} \label{eq:w-definition}
         &= \sum_{j \in [m]} \left(\sum_{i \in L_{g(j),j}} v_{ij} \ell_{ij}\bar y_i^g + \sum_{i \in C_{g(j),j}} v_{ij} c_{ij}\bar y_i^g + v_{g(j)j} \ell_{g(j)j}\bar y_{g(j)}^g\right) \\
         &\phantom{ = } + \sum_{i\in [n]} f_i \bar y^g_i + \sum_{j \in [m]} v^g_{j}\alpha^g_{j}
    \end{split}\\
    & = \sum_{j \in [m]} \sum_{i \in [n]\setminus\{g(j)\}} v_{ij} x_{ij} + \sum_{i\in [n]} f_i y_i + \sum_{j \in [m]} v_{g(j),j} x_{g(j),j}\label{eq:LC-definition}\\
    & = \sum_{j \in [m]} \sum_{i \in [n]} v_{ij} x_{ij} + \sum_{i\in [n]} f_i  y_i, \nonumber
    \end{align}
    where~\eqref{eq:w-definition} follows from the definition of~$w$, and~\eqref{eq:LC-definition} holds because, for each $j \in [m]$ and $i\in [n]\setminus\{g(j)\}$, we have $x_{ij}=\ell_{ij}  y_i$ if $i \in L_{g(j),j}$, $ x_{ij}=c_{ij}  y_i$  if  $i \in C_{g(j),j}$, and $ x_{ij}= \alpha^g_j+\ell_{g(j),j}$ otherwise. Finally, it is clear that any feasible solution of~\eqref{Multi2}--\eqref{Multi9} induces a feasible solution of~\eqref{demand}--\eqref{Integer} with the same value.
    Therefore, \eqref{Multi1}--\eqref{Multi9} is a correct formulation of~\p.

\end{proof}

\noindent This formulation reveals a decomposition of~\p~into a collection of \mkc\ instances.
Precisely, for each~$g \in \mathcal{G}$, the formulation for~$g$ is
\begin{align}
    \min  \;\;  & \sum_{i \in [n]} f^g_i y^g_i + \sum_{j \in [m]} v^g_j \alpha^g_j  & \label{Decomposition1}\\
    \text{s.t.}\;\; & \sum_{i \in [n]} w^g_{ij} y^g_i +  \alpha^g_j\ge d_j   & \forall j \in [m], \\
                & c^g_{g(j),j} \ge \alpha^g_j \ge 0 & \forall j \in [m], \\
                & y^g_{g(j)} = 1 & \forall j \in [m], \\
                & y^g_i \in \{0,1\} & \forall i \in [n].\label{Decomposition2}
\end{align}
Note that, for every $i \in [n]$, $f^g_i$ is the cost of choosing item~$i$, and $w^g_{ij}$ is the weight of item~$i$ in knapsack~$j$, where $j \in [m]$.
For each~$j \in [m]$,  $d_j$ is the demand of  knapsack~$j$, $v^g_j$ and $ c^g_{g(j),j}$ are the cost and the upper bound of the continuous variable associated with knapsack~$j$, respectively.

We conclude this section with the following corollary of Proposition~\ref{prop:union-polyhedral-model}, which is used in the algorithms designed in the next sections.

\begin{corollary}\label{cor:optimal}
    Let $I$ be an instance of \p.
    There exists a collection~$\mathcal{I}$ of~$\bigO(n^{m})$ instances of \mkc\ such that 
    \[\OPT(I) = \min_{I' \in \mathcal{I}} \OPT_{\mkc}(I').\]
\end{corollary}

The formulation developed for~\p~can be naturally adapted to its packing counterpart. By redefining the sets $L_{kj} = \{ i \in [n] \setminus \{k\} : v_{ij} < v_{kj} \} \cup \{ i \in [n] \setminus \{k\} :   v_{ij} = v_{kj} \text{ and } i > k\}$ and $C_{kj} = \{ i \in [n] \setminus \{k\} : v_{ij} > v_{kj}\} \cup \{ i \in [n] \setminus \{k\} :  v_{ij} = v_{kj} \text{ and } i < k\}$, we can decompose the problem into a collection of instances of  multidimensional knapsack packing problem with one continuous variable per dimension, which is denoted by~\mkp.

\section{{A polynomial-time approximation scheme}}\label{sec:ptas}

Next, we design a polynomial-time approximation scheme for~\p~(with $m$ fixed), which is based on the decomposition presented in the previous section.
Frieze and Clarke~\cite{friezeApproximationAlgorithmsMdimensional1984} propose a PTAS for the Multidimensional Knapsack Cover problem without continuous variables. 
In what follows, we extend their PTAS to include exactly one continuous variable per dimension. The analysis of the proposed algorithm follows the methodology of \cite{friezeApproximationAlgorithmsMdimensional1984}.

Let $\eta, \mu \in \Z_>$,  $\bar f \in \Z^\eta_\ge$, $\bar v \in \Z^\mu_\ge$, $\bar c \in \Z^\mu_>$, $\bar w \in \Z^{\eta \mu}_\geq$, and $\bar d \in \Z^\mu_\ge$.
The Multidimensional Knapsack Cover problem with one continuous variable per dimension (\textsc{mkc}) is equivalent to
\begin{align}
    \min  \;\;  & \sum_{i \in [\eta]} \bar f_i y_i + \sum_{j \in [\mu]} \bar v_j \alpha_j  & \label{objective:knapsack-cover} \\
    \text{s.t.}\;\; & \sum_{i \in [\eta]} \bar w_{ij} y_i \ge \bar d_j - \alpha_j   & \forall j \in [\mu], \label{ineq:knapsack-cover-dim}\\
                & \bar c_{j} \ge \alpha_j \ge 0 & \forall j \in [\mu], \\
                & y_i \in \{0,1\} & \forall i \in [\eta].\label{Relaxation}
\end{align}

\noindent Let $\epsilon > 0$, and  let $k = \min\{\eta, \: \lceil \mu /\epsilon\rceil\}$.
We next design an algorithm $\mathcal{A}_\epsilon$ for \mkc\ that runs in polynomial time when $\mu$ (a.k.a. \emph{dimension}) is fixed.
For every $S \subseteq [\eta]$, define~$T(S) = \{i \in [\eta] \setminus S :  \bar f_i > \min_{t \in S} \bar f_t\}$. 
Let $LP(S)$ be the linear relaxation of the formulation obtained from \eqref{objective:knapsack-cover}--\eqref{Relaxation} when replacing~\eqref{Relaxation} by, for all $i \in [\eta]$, $0 \le y_i \le 1$ and 
\begin{equation*}
    y_{i} = 
        \begin{cases} 
        1 & \text{ if } i \in S,\\
        0 & \text{ if } i \in T(S). 
        \end{cases}
\end{equation*}
The proposed algorithm solves $LP(S)$ for every $S \subseteq [\eta]$ such that $|S| \le k$, and then rounds up the corresponding optimal \emph{fractional} solution. 
Finally, it outputs the best solution found. 
This procedure is formally described in Algorithm~\ref{alg:ptas}. 

\begin{algorithm}[tbh!]
  \caption{Algorithm~$\mathcal{A}_\epsilon$ for~\mkc. \label{alg:ptas}}
  \begin{algorithmic}[1]
    \Require{An instance $I=(\eta,\mu,\bar f, \bar v,\bar c, \bar w, \bar d)$ of~\mkc.}
    \Ensure{A feasible solution $(y,\alpha) \in \{0,1\}^{\eta}\times \R^\mu$ to $I$.}
    \Statex
    \Procedure{$\mathcal{A}_\epsilon$}{$I$}
    \Let{$k$}{$\min\{\eta, \: \lceil \mu /\epsilon\rceil\}$}
    \Let {$\zeta$}{$+\infty$}
    \For{$S \subseteq [\eta]$ with $|S| \le k$}\label{line:for-loop}
        \State Let $\bar d_j(S) := \bar d_j - \bar c_{j} - \sum_{i \in S} w_{ij}$  for each $j \in [\mu]$
        \If{$\sum_{i \in [\eta] \setminus (S\cup  T(S))} \bar w_{ij} \ge \bar d_j(S)$   for each $j \in [\mu]$} \Comment{Otherwise, $LP(S)$ is infeasible} \label{line:feasible}
            \State Compute an optimal \emph{basic} solution $(y'(S), \alpha'(S))$ to $LP(S)$ \label{line:lp-solver}
            \Let {$y''_i(S)$}{$\left\lceil y'_i(S) \right\rceil$ for every $i \in [\eta]$} \Comment{Rounding up to an integer solution} 
            \label{line:round-up}
            \Let{$\zeta(S)$}{$\sum_{i \in [\eta]} f_i y''_i(S) + \sum_{j \in [\mu]} \bar v_j \alpha'_j(S)$}
            \If{$\zeta > \zeta(S)$}
                \State \LetIn{$\zeta$}{$\zeta(S)$}, \quad \LetIn{$y$}{$y''(S)$}, \quad \LetIn{$\alpha$}{$\alpha'(S)$}
            \EndIf
        \EndIf
    \EndFor
    \State \Return $(y, \alpha)$

    \EndProcedure
  \end{algorithmic}
\end{algorithm}

\begin{theorem} \label{thm:PTAS-mkc}
    Let $\epsilon >0$.
    Algorithm~\ref{alg:ptas} is a $(1+\epsilon)$-approximation for \mkc\ with fixed dimension.
\end{theorem}
\begin{proof}
    The proof is very similar to the proof in                       \cite{friezeApproximationAlgorithmsMdimensional1984}. We give it here for completeness.
    
    Let $I=(\eta,\mu,\bar f, \bar v,\bar c, \bar w, \bar d)$  be an instance of~\mkc, and let $(y,\alpha)$ be the output of Algorithm~\ref{alg:ptas} on~$I$ (i.e.,~$(y,\alpha)=\mathcal{A}_\epsilon(I)$).
    Line~\ref{line:feasible} guarantees that $LP(S)$ has an optimal solution, say~$(y'(S),\alpha'(S))$, where $S\subseteq [\eta]$ with~$|S|\leq k$.
    By rounding up the entries of $y'(S)$ in line~\ref{line:round-up}, we obtain a vector satisfying constraints~\eqref{ineq:knapsack-cover-dim}--\eqref{Relaxation}.
    Thus, the algorithm produces a solution to~$I$. 
    Let us denote by $\zeta$ the value of~$(y,\alpha)$.
    Consider an optimal solution $(y^*, \alpha^*)$ to~$I$, and denote by $\zeta^*$ the optimal value of~$I$.
    Let us define $S^* = \{i \in [\eta] : y_i^* = 1\}$. 
    If $|S^*| \le k$, then we have $\zeta \le \zeta(S^*) \le \zeta^*$ as $\bar f$ is nonnegative,
    which implies $ \zeta = \zeta^*$. Since the condition $|S^*| \le k$ is always satisfied when $k = \eta$, we henceforth assume $k = \lceil \mu / \epsilon \rceil$.
    If $|S^*| > k$, then consider an ordering of~$S^* = \{i(1), \ldots, i(r)\}$ such that $r:=|S^*|$ and $\bar f_{i(1)} \ge \ldots \ge \bar f_{i(r)}$. 
    We define~$S_k^* = \{i(1), \ldots, i(k)\}$ and $\sigma = \sum_{t \in [k]} \bar f_{i(t)}$. 
    It follows from the definition of $S^*_k$ and $T(S^*_k)$  that $ \bar f_i\le \bar f_{i(k)}$ for each $i \in [\eta]\setminus (S^*_k \cup T(S^*_k))$. 
    Hence, it holds that 
    \begin{equation}\label{proof}
        \bar  f_i \le \frac{\sigma}{k} \text{ for all } i \in [\eta]\setminus (S^*_k \cup T(S^*_k)).    
    \end{equation}
    One may easily check that $(S^*\setminus S^*_k)\cap T(S^*_k) = \emptyset$.
    Thus, $(y^*, \alpha^*)$ is feasible to $LP(S^*_k)$, and so
    \[
    \zeta^* \ge \sum_{i \in [\eta]}\bar  f_i y'_i(S^*_k) + \sum_{j\in [\mu]} \bar v_j \alpha'_j(S^*_k) \ge \sum_{i \in [\eta]} \bar f_i y''_i(S^*_k) - \delta + \sum_{j \in [\mu]} \bar v_j \alpha'_j(S^*_k) \ge \zeta - \delta,
    \]
    where $\delta = \sum_{i \in D} \bar f_i$ and $D = \{i\in [\eta]: 0 < y'_i(S^*_k) < y''_i(S^*_k)\}$.
    Note that every $i \in D$ implies $y_i$ is a basic variable in $y'(S^*_k)$. 
    Thus $ |D| \le \mu$. 
    Moreover, $D \cap (S^*_k \cup T(S^*_k)) = \emptyset$, and so  we have $\bar f_i \le \sigma/ k$ for all $i \in D$ due to~\eqref{proof}. 
    Hence $\delta \le \mu\sigma/k$, and 
    \(
    \zeta^* \ge \zeta - \mu\sigma/ k \ge \zeta - \mu \zeta^*/k\) since \(\zeta^* \ge \sigma\).
    Therefore, 
    \(\zeta \le  (1+\mu/k)\zeta^* \le (1+\epsilon) \zeta^*.\)

    The algorithm's runtime is dominated by solving the linear program on line~\ref{line:lp-solver} for each of the $\bigO(\eta^{k})$ subsets $S$. \cite{cohenSolvingLinearPrograms2021} present the current state-of-the-art method to solve a linear program in $\bigO \left( M(\eta) \log \eta \log \left( \eta / 2^{-\bigO(L)} \right) \right)$ time where $M(\eta)\sim \eta^{2.38}$ is the cost of matrix multiplication and inversion, and $L = \bigO \left( \log(\eta + \norm{\Bar{d}}_{\infty} + \norm{\Bar{f}}_{\infty} + \norm{\Bar{v}}_{\infty})\right)$. A consequence of using their method is that while optimal primal and dual solutions are obtained, an optimal basis is not necessarily identified. Therefore, we apply the algorithm due to Megiddo~\cite{megiddoFindingPrimalDualOptimal1991} to compute an optimal basis in strongly polynomial time from an optimal primal-dual solution pair.

\end{proof}
\begin{theorem} \label{thm:PTAS}
    There exists a polynomial-time approximation scheme for~\p\ when $m$ is fixed.
\end{theorem}
\begin{proof}
Proposition~\ref{prop:union-polyhedral-model} shows a decomposition of~\p\ into $\bigO(n^m)$ instances of the multidimensional knapsack cover problem with a single continuous variable per dimension.
One can run Algorithm~\ref{alg:ptas} on each of these instances and output the cheapest solution.
By Corollary~\ref{cor:optimal} and Theorem~\ref{thm:PTAS-mkc}, this leads to a PTAS for \p\ when $m$ is fixed.

\end{proof}

\noindent Theorem~\ref{thm:PTAS} gives, in a sense, the best polynomial-time approximation one can design to~\p\ since the multidimensional knapsack cover problem -- even without continuous variables -- does not admit any fully polynomial-time approximation scheme  unless $\Pclass = \NP$ 
\citep{magazine1984note}.

\section{A polynomial-time approximation scheme to packing problems} \label{sec:PTAS_Packing}
A PTAS for the packing variant of \p~can be derived with minor modifications to Algorithm~\ref{alg:ptas}, as shown in Algorithm~\ref{alg:ptas_mkp}. The key changes are threefold. First, $\bar d_j(S)$ is redefined as $\bar d_j(S) = \bar d_j - \bar \ell_j - \sum_{i \in S} w_{ij}$ for each $j \in [\mu]$. Second, the feasibility check becomes verifying that $\bar d_j(S) \ge 0$ for each $j \in [\mu]$. Finally, after solving the linear program $LP(S)$, its solutions must be rounded down.

\begin{algorithm}[tbh!]
  \caption{Algorithm~$\mathcal{A}_\epsilon$ for~\mkp. \label{alg:ptas_mkp}}
  \begin{algorithmic}[1]
    \Require{An instance $I=(\eta,\mu,\bar f, \bar v,\bar c, \bar w, \bar d)$ of~\mkp.}
    \Ensure{A feasible solution $(y,\alpha) \in \{0,1\}^{\eta}\times \R^\mu$ to $I$.}
    \Statex
    \Procedure{$\mathcal{A}_\epsilon$}{$I$}
    \Let{$k$}{$\min\{\eta, \: \lceil \mu /\epsilon\rceil\}$}
    \Let {$\zeta$}{$0$}
    \For{$S \subseteq [\eta]$ with $|S| \le k$}\label{line:for-loop_mkp}
        \State Let $\bar d_j(S) := \bar d_j - \bar \ell_{j} - \sum_{i \in S} w_{ij}$  for each $j \in [\mu]$
        \If{$\bar d_j(S) \ge 0$   for each $j \in [\mu]$} \Comment{Otherwise, $LP(S)$ is infeasible} \label{line:feasible_mkp}
            \State Compute an optimal \emph{basic} solution $(y'(S), \alpha'(S))$ to $LP(S)$ \label{line:lp-solver_mkp}
            \Let {$y''_i(S)$}{$\left\lfloor y'_i(S) \right\rfloor$ for every $i \in [\eta]$} \Comment{Rounding down to an integer solution} 
            \label{line:round-up_mkp}
            \Let{$\zeta(S)$}{$\sum_{i \in [\eta]} f_i y''_i(S) + \sum_{j \in [\mu]} \bar v_j \alpha'_j(S)$}
            \If{$\zeta < \zeta(S)$}
                \State \LetIn{$\zeta$}{$\zeta(S)$}, \quad \LetIn{$y$}{$y''(S)$}, \quad \LetIn{$\alpha$}{$\alpha'(S)$}
            \EndIf
        \EndIf
    \EndFor
    \State \Return $(y, \alpha)$

    \EndProcedure
  \end{algorithmic}
\end{algorithm}

The approach previously described can also be extended to yield a PTAS for a related class of problems of the following form, which we denote by~\q:
\begin{align}
        \max  \;\;  &  \sum_{i\in [n]} \sum_{j \in [m]} v_{ij} x_{ij} + \sum_{i \in [n]} \sum_{j \in [m]}f_{ij} y_{ij} & \notag\\
    \text{s.t.}\;\; &  \sum_{i \in [n]} x_{ij} \le d_j  & \forall j \in [m], \label{GAP_packing1}\\ 
    &  \sum_{j \in [m]} y_{ij} \le 1  & \forall i \in [n], \label{GAP_packing2}\\
                & \ell_{ij} y_{ij} \le x_{ij} \le c_{ij} y_{ij} & \forall i \in [n], \; j \in [m], \notag\\
                & y_{ij} \in \{0,1\} & \forall i \in [n], \; j\in [m], \notag
\end{align}
where parameters are defined as before. 

For a fixed $m$, the results from Section~\ref{sec:decomposition} can be adapted to decompose~\q~into 
$\bigO(n^m)$ 
instances of \gap\ with a single continuous variable for each dimension as follows. 
Define $\bar n = n + m$ by introducing $m$ ``dummy'' items~$\{n+1, n+2,  \ldots, n+m\}$. 
For each $i \in  [\bar n]\setminus[n]$ and $j \in [m]$, we set $\ell_{ij} = c_{ij} = v_{ij} = f_{ij} = 0$. 
These newly created items ensure the existence of solutions where no item $i \in  [n]$ is selected for a certain $j \in [m]$. The main difference compared to the approach in Section~\ref{sec:decomposition} lies in the definition of $\mathcal{G}$, here it is the set of all injective functions from $[m]$ to $[\bar n]$. For each $g \in \mathcal{G}$, the formulation for $g$ is: 
\begin{align*}
        \max  \;\;  &  \sum_{i\in [\bar n]} \sum_{j \in [m]} f^g_{ij} y^g_{ij} + \sum_{j \in [m]} v^g_{j} \alpha^g_{j} &\\
    \text{s.t.}\;\; &  \sum_{i \in [\bar n]} w^g_{ij} y^g_{ij} + \alpha^g_j \le d^g_j  & \forall j \in [m],\\
    &  \sum_{j \in [m]} y^g_{ij} \le 1  & \forall i \in [\bar n],\\
                & 0 \le \alpha^g_j \le c^g_{g(j),j}  &  \forall j \in [m],\\
    & y^g_{g(j),j} = 1 & \forall j \in [m], \\
                & y^g_{ij} \in \{0,1\} & \forall i \in [\bar n], j \in [m],
\end{align*}
where $w^g_{ij}, f^g_{ij}, v^g_{j}, d^g_j, c^g_{g(j),j}$ are defined in the same way as in Section~\ref{sec:decomposition}.

Because $m$ is assumed to be fixed, we can transform this \gap\ instance with one continuous variable per dimension into an instance of~\mkp~with some additional cardinality constraints.
We define a new set of $m n + m$ items as follows. For each original item $i \in [n]$ and each $j \in [m]$, we create a new item that has weight $w^g_{ij}$ and profit $f^g_{ij}$ for knapsack $j$, and a weight and profit equal to~$0$ for all other knapsacks $j' \in [m]\setminus\{j\}$. 
We retain the 
$m$ dummy items from the previous augmentation, each with weight $0$ and profit $0$ for all knapsacks. We additionally add a cardinality constraint for each~$i \in [n]$ to~\mkp, namely $\sum_{j \in [m]} y^g_{ij} \le 1$ for all $i \in [n]$.
Note that no cardinality constraints are added for the dummy items. 
We then apply the PTAS presented in Algorithm~\ref{alg:ptas_mkp} to this new instance. The algorithm remains unchanged, except that in the feasibility check (Step~\ref{line:feasible_mkp} of Algorithm~\ref{alg:ptas_mkp}), it additionally ensures that at most one copy of each item $i \in [n]$ is selected. 
This adjustment does not affect the algorithm's analysis. This leads us to the following analogous version of Theorem~\ref{thm:PTAS} for~\q:

\begin{theorem}
    There exists a polynomial-time approximation scheme for~\q~when $m$ is fixed.
\end{theorem}

Furthermore, the same approach extends to the covering version of~\q, yielding a PTAS in the setting where \eqref{GAP_packing1} and \eqref{GAP_packing2} are $\ge$-constraints and the objective function is minimized. Moreover, by replacing $\sum_{j \in [m]} y_{ij} \le 1$  with $\sum_{j \in [m]} y_{ij} \le t_i$ for all $i \in [n]$ for a parameter $t \in [n]^n$, we obtain a more general version of \q\ which also admits a PTAS using the same core method proposed in the previous sections.
Observe that the packing version of~\p\ corresponds to the case where $t_i = n$ for all $i \in [n]$.
Interestingly, the cover version of~\q\ is not a generalization of~\p.

\section{FPTAS and approximate formulation for $m=1$}
\label{sec:fptas}
In this section, we design a \emph{fully} polynomial-time algorithm for~\p\ when~$m = 1$, denoted as \pone.
This problem generalizes the knapsack cover problem as this occurs when $\ell_i = c_i$ for all $i \in [n] $. 
Hence \pone\ is clearly $\NP$-hard.
Theorem~\ref{thm:PTAS} already proves that a PTAS exists for this problem. From the introduction we also know that Van Hoesel and Wagelmans~\cite{vanhoeselFullyPolynomialApproximation2001} present an FPTAS for the single-item capacitated economic lot-sizing problem (which includes the knapsack cover problem) that is more general than \pone~in its consideration of backlogging and holding costs. 
However, their model does not include lower bounds. This distinction in problem scope is accompanied by a methodological divergence as their approach does not rely on a decomposition into knapsack sets. The current state-of-the-art approximation algorithm for the knapsack problem is a randomized pseudo-polynomial algorithm by He and Xu~\cite{qizheng2024simple}, which achieves a time complexity of $\bigO(n^{3/2} \times \min\{w_{\max}, p_{\max}\})$ where $w_{\max}$ represents the maximum weight and $p_{\max}$ the maximum profit of an item. However, it is not obvious how the introduction of a continuous variable affects this result, a question that warrants its own investigation.

The formulation presented in Proposition~\ref{prop:union-polyhedral-model} reveals a decomposition of~\pone~into a collection of~$\bigO(n)$ instances of the knapsack cover problem with a single continuous variable, denoted by~\mkcone.
Consider an arbitrary (but fixed) $g \in [n]$.
As observed in Section~\ref{sec:decomposition}, problem \pone\ for~$g$ is equivalent to an instance~$(\eta, \bar v, \bar f, \bar w, \bar d, \bar c)$ of \mkcone, which can be modeled as
\begin{align*}
        \min  \;\;  & \sum_{i \in [\eta]} \bar f_i y_i + \bar v \alpha &\\
    \text{s.t.}\;\; & \sum_{i \in [\eta]} \bar w_{i} y_i \ge \bar d - \alpha,  & \\
                & \bar c \ge \alpha \ge 0, & \\
                & y_i \in \{0,1\} & \forall i \in [\eta]. 
\end{align*}
\noindent Note that Corollary~\ref{cor:optimal} can also be used in algorithms to solve~\pone. 
In what follows, we propose an FPTAS for~\mkcone~which is outlined in Algorithm~\ref{alg:fptas}. The core of this algorithm is an exact approach that runs in pseudo-polynomial time as outlined in Procedure~\ref{alg:name:dp} of Algorithm~\ref{alg:fptas}. 
This is a dynamic programming algorithm for~\mkcone~without a continuous variable which is essentially the same as the folklore DP for the classic knapsack (packing) problem. 
For the sake of simplicity, this algorithm computes only the value of the solution, however one may  easily modify it to store the solution itself. 
Precisely, procedure~\ref{alg:name:dp} computes, for each $i \in [\eta]$ and $j \in [\sum_{i \in [\eta]} \bar f_i]$: 
\[
M(i,j) = \max \left\{\sum_{s \in S} \bar w_s: S \subseteq \{1,\ldots,i\} \; \text{and} \; \sum_{s \in S} \bar f_s \le j \right\}.
\]
Intuitively, $M$ stands for the maximum amount that can be covered at each cost level $j$ using a subset of the first $i$ elements in~$[\eta]$. 
After computing $M$, Algorithm~\ref{alg:fptas} proceeds to compute the optimal value of the continuous variable. 

\begin{algorithm}[t!]
  \caption{Algorithm~$\mathcal{A}^1_\epsilon$ for \mkcone. \label{alg:fptas}}
  \begin{algorithmic}[1]
    \Require{An instance $I=(\eta,\bar v,\bar f,\bar w,\bar d,\bar c)$ of \mkcone.}
    \Ensure{The value of a feasible solution $(y,\alpha) \in \{0,1\}^\eta \times \R$ of $I$.}
    \Statex
    \Procedure{$\mathcal{A}^1_\epsilon$}{$I$}
    \Let{$I'=(\eta, v',f', \bar w, \bar d,\bar c)$}{\ref{alg:name:scaling}($I$)}\label{line:scaling}
    \Let{$M$}{\ref{alg:name:dp}($I'$)}
    \Let {$q$} {$\infty$}
    \State $j = \sum_{i \in [\eta]} f'_i$
    \While {$M(\eta,j) \ge \bar d - \bar c$ and $j \ge 0$}\label{line:feasibility}
        \If {$j + v'(\max \{ \bar d - M(\eta,j), 0\}) < q$}
            \State $q = j + v'(\max \{\bar d - M(\eta,j), 0\})$
            \State $j = j-1$
        \EndIf
    \EndWhile
    \Let{$(y,\alpha)$}{A feasible solution to~$I'$ of value~$q$}
    \State \Return $(y,\alpha)$
    \EndProcedure
    \Statex
    \Procedure{Scaling}{$I$}
    \funclabel{alg:name:scaling}
        \State $f_{\max} \gets \max_{i \in [\eta]} \bar{f}_i$
        \State $\lambda \gets \max \{\epsilon f_{\max}/(\eta \poly(\eta)); \: 1\}$ \label{line:lambda-def}
        \For{$i \in [\eta]$} \label{line:scaling1}
            \State $f'_i \gets \left\lceil \bar{f}_i/{\lambda} \right\rceil$ \label{line:scaling2}
            
        \EndFor
        \Let{$v'$}{$ \bar v / \lambda $} \label{line:lambda}
        \State \Return $(\eta,v',f',\bar w, \bar d, \bar c)$
    \EndProcedure
    \Statex
    \Procedure{DP}{$I$}
    \funclabel{alg:name:dp}
    \State $M(i,0) = 0$ for all $i \in [\eta]$
    \For{$i \in [\eta]$}
        \For{$j$  from $0$ to $\sum_{i \in \eta}  \bar f_i$}
            \If {$ \bar f_i > j$} 
                \Let {$M(i,j)$} {$M(i-1,j)$}
            \Else 
                \Let {$M(i,j)$} {$\max \{M(i-1, j); \: \bar w_i + M(i-1, j- \bar f_i)\}$}
            \EndIf
        \EndFor
    \EndFor
    \State \Return $M$
    \EndProcedure
  \end{algorithmic}
\end{algorithm}
\begin{theorem} \label{theorem:FPTAS}
    Let $I = (\eta, \bar v, \bar f, \bar w, \bar d, \bar c)$ be an instance of \mkcone\ with $f_{\max} \ge f_{\min} \ge  f_{\max}/\poly(\eta)$, where~$f_{\max} = \max_{i \in [\eta]} \bar f_i$ and $f_{\min} = \min\left\{ \min_{i \in [\eta]} \bar f_i, \: \bar v \bar d \right\}$, and $\poly(\eta)\ge 1$ is a polynomial function on $\eta$.
    Algorithm~\ref{alg:fptas} is a $(1+\epsilon)$-approximation that runs in 
    $\bigO\left(\eta^3 \poly(\eta)/\epsilon\right)$ time, for any $\epsilon>0$. 
\end{theorem}
\begin{proof}
    We first analyze the algorithm supposing that no scaling is performed in line~\ref{line:scaling}, that is, $I'=I$.
    Let $q$ be the value produced by $\mathcal{A}^1_\epsilon$ on input $I$.
    It follows from the computation of $M$ by procedure~\ref{alg:name:dp} and the conditions in line~\ref{line:feasibility} that $q$ corresponds to the value of a feasible solution of~$I$.
    Let $(y^*, \alpha^*)$ be an optimal solution for~$I$ with $\alpha^*$ integer.
    The existence of this solution is guaranteed since $\bar w, \bar d$ and $\bar c$ are all integer. 
    Let $j^* = \sum_{i\in [\eta]}  \bar f_i y^*_i$. Note that $M(\eta,j^*) \ge \sum_{i \in [\eta]} \bar w_i y^*_i \ge \bar d - \alpha^*$. 
    Hence, $\alpha^* \ge \bar d - M(n,j^*)$. 
    Therefore, we have $q \le j^* + \bar v(\bar d - M(n,j^*)) \le \sum_{i \in [\eta]} \bar f_i y^*_i + \bar v \alpha^*=\OPT(I)$.
    Clearly, the running time of this exact algorithm depends on $\bar f$.
    This explains why we first scale down the input using procedure~\ref{alg:name:scaling} in line~\ref{line:scaling} before executing the pseudopolynomial-time algorithm to solve the problem.
    As we shall see, this transformation guarantees that the algorithm runs in polynomial time on the size of $I$ if $f^{\max}/f^{\min}$ is bounded by a polynomial function on $\eta$. Let $I'$ denote the transformed instance in line~\ref{line:scaling}.
    Let~$q$ be the output of Algorithm~\ref{alg:fptas} on input $I$, and let~$(y,\alpha) \in \{0,1\}^\eta \times \R$ be a feasible solution of $I$ of value~$q$.
    By the construction, $(y,\alpha)$ is optimal for~$I'$.
    Note that if~$\lambda=1$, then this solution is also optimal for $I$.
    Hereafter, assume that $\lambda=\epsilon  f_{\max}/(\eta \poly(\eta))>1$.
    We define $Q=\{ i \in [\eta] : y_i=1\}$ and 
    $Q^*=\{ i \in [\eta] : y^*_i=1\}$. 
    First, note that $(y^*,\alpha^*)$ is also a feasible solution for $I'$, and so 
    $\sum_{i \in Q^*} f'_i +   v' \alpha^* \ge \sum_{i \in Q} f'_{i} + v' \alpha =\OPT(I')$.
    It is clear from lines~\ref{line:lambda-def}--\ref{line:lambda} that $\frac{\bar f_i}{\lambda} \le f'_i \le \frac{\bar f_i}{\lambda} + 1$ for every $i\in [\eta]$ and $\frac{\bar v}{\lambda} = v'$. Hence, the following sequence of inequalities hold for the solution $(y, \alpha)$ produced by Algorithm~\ref{alg:fptas}:
    \begin{align*}
        \sum_{i \in Q} \bar f_i + \bar v \alpha &  \le \lambda \sum_{i \in Q} f'_i + \lambda  v' \alpha = \lambda \OPT(I')\\
        & \le \lambda \sum_{i \in Q^*} f'_i + \lambda v' \alpha^* \le \lambda \sum_{i \in Q^*} \left( \frac{\bar f_i}{\lambda} + 1 \right) + \lambda \left( \frac{\bar v}{\lambda}\right) \alpha^* \\
        & = \OPT(I) + \lambda |Q^*| \le \OPT(I) + \dfrac{\epsilon  f_{\max}}{ \poly(\eta)} \\
        & \leq (1+\epsilon)\OPT(I),
    \end{align*}
    where the last inequality holds because $\OPT(I) \ge f_{\min} \ge f_{\max}/\poly(\eta)$.
    
    The running-time of Algorithm~\ref{alg:fptas} on input $I'$ takes $\bigO(n^2 f'_{\max})$ time , with $f'_{\max} = \max_{i \in [\eta]} f'_i$. 
    It follows from the scaling in lines~\ref{line:lambda-def}-\ref{line:scaling2} that $f'_{\max} \le f_{\max}/\lambda + 1 \le \eta \poly(\eta)/\epsilon + 1$. Therefore, Algorithm~\ref{alg:fptas} runs in~$\bigO\left(\left(\eta \poly(\eta)/\epsilon + 1\right) \eta^2\right)$ time.

\end{proof}
\noindent The hypothesis in Theorem~\ref{theorem:FPTAS} on the polynomial bound on the ratio $f_{\max}/f_{\min}$ is due to fact that the optimal value of an instance of \mkcone\  can be arbitrarily smaller than $f_{\max}$ in general.
As an example, consider an instance of \mkcone\ with $\eta=2$, $\bar d =1$, $f_1=2$, $w_1=1$, $f_2=100$, $w_2=100$, and $\bar v=100$. Here, $f_{\max} = 100$ while the optimal value equals~$2$. A key direction for future research is to eliminate the assumption on the polynomial bound on the ratio $f_{\max}/f_{\min}$. We hypothesize that an approach similar to that of \cite{bienstockTighteningSimpleMixedinteger2012} could be adapted for this purpose, but formally establishing this remains an open challenge.

For problem \pone, Theorem~\ref{theorem:FPTAS} yields the following result. 

\begin{theorem}
There exists a fully polynomial-time approximation scheme for \p~when $m=1$, assuming $f_{\max} \ge f_{\min} \ge  f_{\max}/\poly(n)$, where~$f_{\max} = \max_{i \in n} f_i + c_i v_i$ and $f_{\min} = \min_{i \in n} f_i + \ell_i v_i$, and $\poly(n)\geq 1$ is any polynomial function on $n$. 
\end{theorem}
\begin{proof}
    Proposition~\ref{prop:union-polyhedral-model} shows a decomposition of \pone~into $\bigO(n)$ instances of the knapsack cover problem with a single continuous variable.
    One can run Algorithm~\ref{alg:fptas} on each of these instances and output the cheapest solution.
    By Corollary~\ref{cor:optimal} and Theorem~\ref{theorem:FPTAS}, this leads to an FPTAS for \pone~when for every resulting knapsack cover problem with a single continuous variable, $f_{\max} \ge f_{\min} \ge  f_{\max}/\poly(n)$, where~$f_{\max} = \max_{i \in [\eta]}  \bar f_i$ and $f_{\min} = \min\left\{ \min_{i \in [\eta]}  \bar f_i, \:  \bar v  \bar d \right\}$, and $\poly(\eta)\geq 1$ is any polynomial function on $\eta$. It is clear that this is guaranteed by the assumption on the input data of~\p.  
    The total running time of this algorithm is~$\bigO(n^4 \poly(n)/\epsilon)$.

\end{proof}

By utilizing a similar algorithm as~$\mathcal{A}^1_\epsilon$, it is possible to obtain an analogous approximation result for the packing version of~\pone. Furthermore, in the packing case, the assumption of a polynomially bounded ratio $f_{\max}/f_{\min}$ is unnecessary, 
as the optimal value is at most $n f_{\max}$.

An important consequence of Theorem~\ref{thm:PTAS} and the equivalence of optimization and separation is that an $\epsilon$-approximate formulation must exist for \pone. Since Proposition~\ref{prop:union-polyhedral-model} decomposes \pone~into $\bigO(n)$ instances of~\mkcone, we can directly leverage the results of \cite{bienstockTighteningSimpleMixedinteger2012}, who obtain a provably tight, polynomially large formulation for the standard knapsack cover problem. 
We extend their methodology to incorporate a continuous variable, defining a linear programming formulation whose optimal value is at most $(1+\epsilon)$ times the optimal value of~\mkcone.

\noindent Leveraging this building block, we immediately obtain the following result for \pone.

\begin{theorem}  
    Let $\epsilon \in (0,1)$, and let $I$ be an instance of~\pone. 
    There is a linear programming formulation~$A^Iz\le b^I$ with 
    $\bigO\left(\epsilon^{(\log \epsilon/ \log (1+\epsilon)) - 1}n^3 \right)$
    variables and constraints, and objective function~$h^I$ such that
    \[\OPT(I) \le (1+\epsilon)\min\{h^I(z) :  A^Iz\le b^I\}.\]
\end{theorem}
\begin{proof}[Proof sketch]
    Proposition~\ref{prop:union-polyhedral-model} shows a decomposition of \pone~into $\bigO(n)$ instances of the knapsack cover problem with a single continuous variable. 
    For each of these problems, one may create an $\epsilon$-approximate formulation with $\bigO\left(\epsilon^{(\log \epsilon/ \log (1+\epsilon)) - 1}n^2 \right)$
    variables and constraints, and combine them as in the formulation presented in Section~\ref{sec:decomposition}.
    This leads to an $\epsilon$-approximate formulation for \pone~with 
    $\bigO\left(\epsilon^{(\log \epsilon/ \log (1+\epsilon)) - 1}n^3 \right)$
    variables and constraints.

\end{proof}

The proof of the $\epsilon$-approximate formulation for the knapsack cover problem with a single continuous variable can be found in the appendix as the approach and structure adheres closely to what is shown by Bienstock and McClosky~\cite{bienstockTighteningSimpleMixedinteger2012}. 

The existence of an $\epsilon$-approximate formulation for the packing version of~\pone~is guaranteed in principle by the equivalence of optimization and separation. However, such a formulation is not immediately available. A key complication arises from the presence of a continuous variable after the decomposition of~\pone. While Bienstock~\cite{bienstockApproximateFormulations012008} successfully derived an approximate formulation for the standard 0-1 knapsack set, it is unclear how to adapt the techniques to handle an additional continuous variable. Resolving this gap, and thus developing an explicit $\epsilon$-approximate formulation for the packing variant of~\pone, is a significant open challenge.

\section{{Perfect compact formulation for $m = 1$ and uniform bounds}} \label{Section:extended}
In this section, we present a compact and tight formulation for~\p~when $m = 1$ and the bounds are uniform, that is, $\ell_{i} = \ell$ and $c_{i} = c$ for all $i \in [n]$. Note that each item $i \in [n]$ still has objective function values $v_i$ and $f_i$.
This special case is denoted as~\pu. \pu~can be solved in polynomial time~\citep{hellionPolynomialTimeAlgorithm2012a}. 
However, no tight formulation of the problem is known. 
Hence the proposed perfect compact formulation contributes in this direction.

Constantino~\cite{constantinoLowerBoundsLotSizing1998a} shows that solving~\pu, where $\ell_i$ can vary with $i \in [n]$, is in general $\NP$-hard. This paper also analyzes a setting identical to~\pu~and obtains two families of valid inequalities. Using both families, it is possible to describe the convex hull when $ c \ge d$. \cite{agraPolyhedralDescriptionInteger2006a} study the case where the $y$ variables are general integers and unbounded.

Dash et al.~\cite{dashContinuousKnapsackSet2016a} study the convex hull of the continuous knapsack set. Among other results, they prove that the  convex hull with two integer and one bounded continuous variables is described using inequalities for a  knapsack set with two integer variables and  the ones for a continuous knapsack set with two integer and one continuous variables that is unbounded. We use a special case of this result below. 

We first present a result on the the convex hull of a simpler set. For two integers $a$ and $b$, we use $[a,b]$ to denote the set of integers $\{a, \ldots, b\}$ if $a\leq b$; otherwise $[a,b]=\emptyset$.

\begin{lemma} \label{lemma:convhull}
 Let $$Y=\{(\alpha, \psi)\in \R_{\ge} \times \{0,1\}^{\nu}: \alpha + \sum_{i\in [\nu]} \psi_i \geq \delta, \alpha \leq \sigma\},$$ where $\delta$ and $\nu$ are positive, $\nu$ is an integer such that $\nu\geq \lceil \delta - \sigma \rceil$ and $\sigma\leq 1$. 
 The description of the convex hull of~$Y$ is given by the trivial inequalities,
\begin{align}
 & \sum_{i\in [\nu]} \psi_i \geq \lceil \delta - \sigma \rceil,\label{y1} \text{ and }\\
 & \alpha \geq (\delta - \lfloor \delta \rfloor) \left (\lceil \delta \rceil - \sum_{i\in [\nu]} \psi_i \right). \label{y4} 
\end{align}
   \end{lemma}
\begin{proof}
By Theorem 4.4.~in \cite{ dashContinuousKnapsackSet2016a}, we know that $\conv\{(\alpha, \gamma)\in \R_{\ge} \times \Z_{\ge}: \alpha + \gamma \geq \delta, \alpha \leq \sigma\}$ is the intersection of three sets, namely, $\conv\{(\alpha, \gamma)\in \R_{\ge} \times \Z_{\ge}: \alpha + \gamma \geq \delta\}$, $\{(\alpha, \gamma)\in \R\times \R: \alpha \leq \sigma\}$ and $\conv\{(\alpha, \gamma)\in \R \times \Z_{\ge}:  \gamma \geq \delta - \sigma\}$. In the nontrivial case, the mixed integer rounding inequality $\alpha \geq (\delta - \lfloor \delta \rfloor)(\lceil \delta \rceil - \gamma)$ is needed to describe the first set, and the integer rounding inequality $\gamma \geq \lceil \delta - \sigma \rceil$ is needed to describe the third set \citep[see, e.g.,][]{wolsey2020integer}. Adding an upper bound of $\nu$ on the integer variable $\gamma$ does not create fractional extreme points. 

Now we know that in every extreme point of the set defined by the trivial inequalities and the valid inequalities \eqref{y4} and \eqref{y1}, $\sum_{i = 1}^{\nu} \psi_i$ is integer. If $\psi$ is fractional, then it has at least two entries that are fractional and  such a point cannot be an extreme point.

\end{proof}
 
A perfect formulation for~\pu~is obtained by integrating the decomposition into knapsack cover problems with one continuous variable presented in Section~\ref{sec:decomposition}, and the union-of-polyhedra approach to the convex hull of sets described in Lemma~\ref{lemma:convhull}.

\begin{theorem} \label{theorem:perfect_formulation}
    \pu~admits the following perfect formulation with $\bigO(n^3)$ variables and $\bigO(n^3)$ constraints:

\begin{align*}
        \min  \;\;  &  \sum_{i \in [n]} v_{i} x_{i} + \sum_{i \in [n]} f_i y_i  \\
    \text{\rm s.t.} \;\; & x_{i}=\sum_{g =i+1}^n \sum_{b\in B^g} \ell y_i^{gb} + \sum_{b \in B^i} x_i^{ib} + \sum_{g=1}^{i-1}\sum_{b\in B^g} c y_i^{gb}  & \forall i \in [n], \\
                & y_{i}=\sum_{g \in \mathcal{G}} \sum_{b\in B^g} y_i^{gb}  &  \forall i \in [n], \\
             & \sum_{i=g+1}^n y^{gb}_i \ge \left\lceil \frac{d - b \ell - c}{c}\right\rceil z^{gb} & \forall g \in \mathcal{G}, b \in B^g, \\
               & x_g^{gb} +  \sum_{i=g+1}^n \left(d - (b+1) \ell - \left\lfloor \frac{(d - (b+1) \ell}{c}\right\rfloor c \right)             y^{gb}_i \geq & \\
               & \left( \ell + \left\lceil \frac{d - (b+1) \ell}{c}\right\rceil \times \left(d - (b+1) \ell - \left\lfloor \frac{(d - (b+1) \ell}{c}\right\rfloor c \right) \right) z^{gb}  &  \forall g \in \mathcal{G}, b \in B^g, \\
               & x_g^{gb} + \sum_{i=g+1}^n c y^{gb}_i \ge (d - b \ell)z^{gb} &  \forall g \in \mathcal{G}, b \in B^g, \\
                & \sum_{g \in \mathcal{G}} \sum_{b \in B^g} z^{gb} = 1,   \\
                & \sum_{i=1}^{g-1} y_i^{gb} = b z^{gb} &  \forall g \in \mathcal{G}, b \in B^g, \\
                & y_{g}^{gb} = z^{gb} &  \forall g \in \mathcal{G}, b \in B^g, \\
                & x_g^{gb} \ge \ell z^{gb} &  \forall g \in \mathcal{G}, b \in B^g, \\
                & x_g^{gb} \le c z^{gb} &  \forall g \in \mathcal{G}, b \in B^g, \\
                & y_i^{gb} \le z^{gb} &  \forall i \in [n], g \in \mathcal{G}, b \in B^g, \\
                & y_i^{gb} \ge 0 &  \forall i \in [n], g \in \mathcal{G}, b \in B^g, \\
             %   & z^{gb} \le 1 &  \forall g \in \mathcal{G}, b \in B^g, \\
                & z^{gb} \ge 0 &  \forall g \in \mathcal{G}, b \in B^g, 
\end{align*}
where $[n]$ is ordered so that $v_{1} \ge v_{2} \ge \ldots \ge v_{n}$, $\mathcal{G}=\left [1, \min \left \{\left\lfloor \frac{(n+1)c -d - \ell}{c - \ell} \right\rfloor,  n\right \}\right ]$, $B^g=\emptyset$ for all $g\in [n]\setminus \mathcal{G}$ and $B^g=\left [\max \left\{0,\left\lceil \frac{(d- (n+1-g)c) }{\ell} \right\rceil \right\}, g-1\right ]$  for $g\in \mathcal{G}$.  
\end{theorem}
\begin{proof}
First, $\min \left \{\left\lfloor \frac{(n+1)c -d - \ell}{c - \ell} \right\rfloor,  n\right \}$ is the largest index of a fractional variable that leads to a feasible solution, $\mathcal{G}=\left [1, \min \left \{\left\lfloor \frac{(n+1)c -d - \ell}{c - \ell} \right\rfloor,  n\right \}\right ]$ is the set of possible fractional variables, $\max \left\{0,\left\lceil \frac{(d- (n+1-g)c) }{\ell} \right\rceil \right\}$ is the minimum value of $\sum_{i=1}^{g-1} y_i$ that leads to a feasible solution and $B^g=\left [\max \left\{0,\left\lceil \frac{(d- (n+1-g)c) }{\ell} \right\rceil \right\}, g-1\right ]$ is the set of possible values for $\sum_{i=1}^{g-1} y_i$ at solutions of $X^g$ for $g\in \mathcal{G}$. 

Note that \[\conv(X)= \conv \bigcup_{g\in \mathcal{G}}\bigcup_{b\in B^g} \conv \left (X^g\cap \left \{(x,y): \sum_{i=1}^{g-1}  y_i=b \right \} \right),\] 
where $X^g$ is the set of all points  $(x,y)$ that satisfy 
\begin{align*}
& x_i = \ell y_i \quad i\in [g-1],\\
& x_i = cy_i \quad i\in [g+1, n],\\
%& x_g = \alpha\\
& \sum_{i=1}^{g-1} \ell y_i +x_g + \sum_{i=g+1}^n c y_i \ge d, & \\
&   \ell\leq x_g \leq c, & \\
& y_g=1, & \\
&    y_i \in \{0,1\} \quad  \forall i \in [n].
\end{align*}

For $g\in \mathcal{G}$ and $b\in B^g$, we have  
\begin{align*} %\label{intersection_conv_hull}
\conv \left(X^g\cap \left\{(x,y) \in \R^{n}_\ge\times \{0,1\}^{n}: \sum_{i=1}^{g-1} y_i=b \right\}\right) & =\conv(X^{g1}(b))\cap \conv(X^{g2}(b))\notag \\
&\phantom{=} \cap
\{(x,y) \in \R^{n}\times \R^{n}: x_i = \ell y_i \: \forall i\in [g-1], \notag \\
& \phantom{=\cap \{(x,y) \in \R^{n}\times \R^{n}: }\:\: x_i = cy_i \:\forall i\in [g+1,  n]\}, 
\end{align*}
where 
$$X^{g1}(b)=\left \{(x,y) \in \R^{n}\times \{0,1\}^{n}: \sum_{i=1}^{g-1} y_i=b, y_g=1\right \}$$
and 
$$X^{g2}(b)=\left \{(x,y)\in \R^{n}_\ge\times \{0,1\}^{n}: x_g + \sum_{i=g+1}^n c y_i \ge d-\ell b, \ell\leq x_g \leq c \right \}.$$ 
\\
The description of $\conv(X^{g1}(b))$ is given by 
\begin{align*}
    & \sum_{i=1}^{g-1} y_i = b,& \\
    & y_g=1, &\\
    & 0 \leq y_i \leq 1 \quad \forall i \in [g-1], 
\end{align*}
and the description of $\conv(X^{g2}(b))$ is given by 
\begin{align*}
 & x_g + \sum_{i=g+1}^n  c y_i \ge d - b \ell, &\\
& \ell \leq x_g \le c,\\
  & 0 \leq y_i \le 1 \quad \forall i \in [g+1, n],\\
  &    \sum_{i=g+1}^n y_i \ge \left\lceil \frac{d - b\ell- c}{c}\right\rceil, & \\
 &    x_g  \geq \ell + \left(\frac{d - (b+1) \ell}{c} - \left\lfloor \frac{d - (b+1) \ell}{c}\right\rfloor  \right) \left(\left\lceil \frac{d - (b+1) \ell}{c}\right\rceil -\sum_{i=g+1}^n y_i\right), 
  \end{align*}
after setting $\nu=n-g$, $\delta=\frac{d-(b+1)\ell}{c}$,  $\sigma=\frac{c-\ell}{c}$ and $\alpha=\frac{x_g-\ell}{c}$ in Lemma~\ref{lemma:convhull}. 
Using the result on the convex hull of the union of polyhedra \citep{conforti2008compact, balas1998disjunctive}, we obtain the tight formulation.

\end{proof}

A perfect compact formulation can also be derived for the packing version of~\pu. This is achieved by using sets $\mathcal{G}= [n]$ and $B^g=\left [0, \min \{ \left\lfloor d - \ell/\ell\right\rfloor, g \} \right]$ for $g\in \mathcal{G}$. The primary distinction from the covering case lies in the description of the convex hull $\conv(X^{g2}(b))$, which can be characterized using an approach analogous to Lemma~\ref{lemma:convhull}.
\section{Conclusion and further research}\label{sec:conclusion}
This paper presents approximation schemes for various mixed-integer problems with a fixed number of constraints. By analyzing the structure of the extremes points of the associated polytope of such problems, we decompose it into multidimensional knapsack cover instances with a single continuous variable per dimension. 
By leveraging approximation algorithms from the knapsack literature, we obtain a PTAS for~\p. Utilizing a similar methodology, an FPTAS and an $\epsilon$-approximate formulation is presented for the one-dimensional case. 
For the one-dimensional case with uniform bounds, we propose a perfect compact formulation. To obtain these results, some classical knapsack results are extended to incorporate continuous variables. These results are especially valuable as they provide a framework to design algorithms and tight formulations for similarly structured problems.
We also derive analogous results for the packing and more general variants of \p.

Several promising directions for future research emerge from this work. 
A natural first direction is to extend the results  to more general versions of \p\ with bounded integer variables and piecewise-linear objective functions. 
Also, obtaining an $\epsilon$-approximate formulation for~\p~presents an interesting avenue for further exploration.

\bibliography{Bibliography}

\newpage
\appendix
\section{Appendix}
Define~\mkcone:
\begin{align*}
        \OPT^Z =\;\;  &  \min \bar v \alpha  + \sum_{i \in [\eta]} \bar f_i y_i  &\\
    \text{s.t.}\;\; & \sum_{i \in [\eta]} \bar w_{i} y_i \ge \bar d - \alpha  & \\
                & \bar c \ge \alpha \ge 0 & \\
                & y_i \in \{0,1\} & \forall i \in [\eta],
\end{align*}
where $\bar f$ and $\bar v$ are both positive, $0 < \bar w_i \le \bar d$ for all $i \in [\eta]$, and $0 < \bar c \le \bar d$. Let $OPT^*$ be the value of the LP relaxation of~\mkcone. 
Assume, without loss  of generality, that  $\bar f_1 \ge \bar f_2 \ge \ldots \ge \bar f_{\eta}$. 
In this appendix, we prove the following theorem using an approach and structure that adheres closely to what is shown by \citet{bienstockTighteningSimpleMixedinteger2012}.
\begin{theorem}
    Let $\epsilon \in (0,1)$. 
    There exists a linear programming relaxation of~\mkcone~with $\bigO \left((\frac{1}{\epsilon})^{\bigO(\frac{1}{\epsilon^2})}\eta^2\right)$ variables and constraints whose optimal value $\OPT(\epsilon)$ satisfies $\OPT^Z < (1+\epsilon)\OPT(\epsilon)$. 
\end{theorem} 
\noindent Before proving this result, we first show the following lemma.
\begin{lemma} \label{lemma:appendix1}
    Let $H \ge 1$ be an integer. 
    Suppose $S' \subseteq [\eta]$, and let $0 \le \Bar{y_i} \le 1 $ $(i \in S')$ be given values. 
    Let $f^{\max} = \max_{i \in S'}\{\bar f_i\}$, $f^{\min} = \min_{i \in S'}\{\bar f_i\}$. 
    \\
    \\
    (a) Suppose first that 
    \[
    \sum_{i \in S'} \Bar{y_i} = H
    \]
    Then there exist 0/1 values $\Hat{y_i}$ ($i \in S'$) such that
 \begin{equation}
     \sum_{i \in S'} \bar w_i \Hat{y_i} \ge \sum_{i \in S'} \bar w_i \Bar{y_i} \label{FPTAS_formulation}
 \end{equation}
    and 
    \[
    \sum_{i \in S'} \bar f_i \Hat{y_i} \le (1 - \frac{1}{H} + \frac{f^{\max}}{H f^{\min}})\sum_{i \in S'} \bar f_i \Bar{y_i}
    \]
    (b) Suppose next that
    \[
    \sum_{i \in S'} \Bar{y_i} \ge H
    \]
    Then there exist 0/1 values $\Hat{y_i}$ ($i \in S'$) satisfying \eqref{FPTAS_formulation} and 
    \[
    \sum_{i \in S'} \bar f_i \Hat{y_i} \le (1 + \frac{f^{\max}}{H f^{\min}})\sum_{i \in S'} \bar f_i \Bar{y_i}
    \]
\end{lemma}
\begin{proof}
    (a) Let $\dot{y}$ be an extreme point solution to the linear program
    \begin{align}
        \min\;\;  &  \sum_{i \in S'} \bar f_i y_i  & \notag\\
    \text{s.t.}\;\; & \sum_{i \in S'} \bar w_{i} y_i \ge \sum_{i \in S'} \bar w_i \Bar{y_i}    & \label{28} \\
                &\sum_{i \in S'} y_i = H & \label{29}\\
                &0 \le y_i \le 1 & \forall i \in S' \notag
\end{align}
Because this linear program only has two constraints (excluding variable bounds), any basic feasible solution (extreme point) has at most two basic variables. 
Consequently, at most two of the values $\dot{y_i}$, $i \in S'$, are fractional, but since $H$ is integral, either zero or exactly two $\dot{y_i}$ are fractional. 
Suppose there exist $j,k \in S'$ with $j\neq k$ such that
\[
0 < \dot{y_j} < 1, \quad 0 < \dot{y_k} < 1, \quad \dot{y_j} + \dot{y_k} = 1.
\]
Assume, without loss of generality, that $\bar w_j \ge \bar w_k$. Then we can set $\Hat{y_j} = 1$, $\Hat{y_k} = 0$, and $\Hat{y_i} = \dot{y_i}$ for all $i \in S'\setminus\{j,k\}$, thereby obtaining a 0/1 vector $\Hat{y}$ which satisfies \eqref{28} while increasing the cost of $\dot y$ by at most 
\[
\bar f_j - \bar f_j \dot{y_j} - \bar f_k \dot{y_k} = (\bar f_j - \bar f_k) (1 - \dot{y_j}) \le f^{\max} - f^{\min}.
\]
Hence we obtain
\[
\frac{\sum_{i \in S'} \bar f_i \Hat{y_i} - \sum_{i \in S'} \bar f_i \dot{y}_i}{\sum_{i \in S'} \bar f_i \dot{y}_i} \le \frac{f^{\max} - f^{\min}}{\sum_{i \in S'} \bar f_i \dot{y}_i} \le \frac{f^{\max} - f^{\min}}{H f^{\min}} 
.\]
\\
\\
(b) Proceeding in a similar way to (a) (using, instead of \eqref{29}, $\sum_{i \in S'} y_i \ge H $), it can be assumed that either zero, one or two of the $\dot{y_i}$ ($i \in S'$) are fractional. If two are fractional the result is implied by (a). If there is only one fractional $\dot{y_i}$ then rounding up $\dot{y}$ provides a 0/1 vector $\Hat{y}$ that is feasible while increasing the cost by at most $f^{\max}$. Hence it holds that
\[
\frac{\sum_{i \in S'} \bar f_i \Hat{y_i} - \sum_{i \in S'} \bar f_i \dot{y}_i}{\sum_{i \in S'} \bar f_i \dot{y}_i} \le \frac{f^{\max}}{H f^{\min}}
.\]
\end{proof}

Let $\epsilon \in (0,1)$. Define $K$ as the smallest integer such that $(1 + \epsilon)^{-K} \le \epsilon$. Note that $K \le 1 +\log(1/\epsilon)/\log (1+\epsilon) $.
Let us define $J = \lceil 1 + \frac{1}{\epsilon} \rceil$. 
\begin{definition}
A signature~$\sigma$ is a vector in $\Z^K$ such that $0 \le \sigma_i \le J$ for $i \in  [K]$.
\\
\\
For every $h \in [\eta]$ and  $k \in [K]$, define 
\[
S^{h,k} = \left\{i \in [\eta]: \bar f_h (1 + \epsilon)^{-(k-1)} \ge \bar f_i > \bar f_h (1 + \epsilon)^{-k}  \text{ and } i > h \right\}. 
\]
For each $h \in [\eta]$ and each signature $\sigma$, define 
\begin{equation*}
\begin{split}
P^{h,\sigma} = \Bigg\{ (y,\alpha) \in [0,1]^{\eta} \times \mathbb{R} : & \sum_{i \in [\eta]} \bar w_i y_i \ge \bar d - \alpha, \quad \bar c \ge \alpha \ge 0, \\
& y_1 = y_2 = \ldots = y_{h-1} = 0, \quad y_h = 1, \\
& \sum_{i \in S^{h,k}} y_i = \sigma_k \quad \forall k: \sigma_k < J, \\
& \sum_{i \in S^{h,k}} y_i \ge J \quad \forall k: \sigma_k = J \Bigg\}
\end{split}
\end{equation*}
\end{definition}
\noindent Note that the sets $S^{h,k}$ partition the variables $y_i$ whose cost $\bar f_i$ lies between $\bar f_h$ and $\epsilon \bar f_h$. Furthermore, all variables in each set $S^{h,k}$ have nearly the same cost. 
For a given $P^{h,\sigma}$, the signature $\sigma$ counts the number of $y_i$ that take value $1$ in each $S^{h,k}$ with $k \in [K]$. 
Thus every feasible solution for~\mkcone~belongs to some $P^{h,\sigma}$. To see why, consider a feasible solution of~\mkcone~and let $h$ be the smallest index $i \in [\eta]$ for which $y_i = 1$. For any set $S^{h,k}$, let $\ell$ be the number of indices $i \in [\eta]$ within this set for which $y_i = 1$. The value of $\sigma_k$ is then defined as $\min\{\ell, J\}$.

\begin{lemma} \label{lemma:appendix2}
    For each $h \in [\eta]$ and signature~$\sigma$ with $P^{h,\sigma} \neq \emptyset$, there is $(\Hat{y}^{h,\sigma}, \Hat{\alpha}^{h,\sigma}) \in \B^{\eta}\times \R_\ge $ which is feasible for~\mkcone~such that $\bar f^T \Hat{y}^{h,\sigma} + \bar v \Hat{\alpha}^{h,\sigma} \le (1 + \epsilon) \min \{\bar f^T y + \bar v \alpha : (y,\alpha) \in P^{h,\sigma}\}$. As a result, $\OPT^Z \le (1 + \epsilon) \min \{\bar f^T y + \bar v \alpha: (y,\alpha) \in  P^{h,\sigma}\}$.
\end{lemma}
\begin{proof}
    Let $(\Bar{y},\Bar{\alpha}) \in P^{h,\sigma}$. Define $f^{\max} = \max_{i \in S^{h,k}} \{ \bar f_i\}$, $f^{\min} = \min_{i \in S^{h,k}} \{ \bar f_i\}$.
    Let $\Hat{y}^{h,\sigma}$ be defined as follows. First, for each $k \in [K]$ such that $\sigma_k > 0$, we obtain the values $\Hat{y}_i^{h,\sigma}$ for each $i \in S^{h,k}$ by applying Lemma~\ref{lemma:appendix1} with $S' = S^{h,k}$; note that when $\sigma_k < J$, then we have 
    \[
    1 - \frac{1}{\sigma_k} + \frac{f^{\max}}{\sigma_k f^{\min}} \le \left(1 - \frac{1}{\sigma_k}\right) \frac{f^{\max}}{f^{\min}}+ \frac{f^{\max}}{\sigma_k f^{\min}} = \frac{f^{\max}}{f^{\min}} \le 1 + \epsilon
    \] 
    \noindent by construction of the sets $S^{h,k}$. If $\sigma_k = J$, we have 
\[
1 + \frac{f^{\max}}{J f^{\min}} \le 1 + \epsilon
\]
by our choice for $J$. 
If on the other hand $ \sigma_k = 0$ we set $\Hat{y}_i^{h,\sigma} = 0$ for every $i \in S^{h,k}$. Finally, define $T^h = \{i \in [\eta]: \bar f_i \le (1 + \epsilon)^{-K} \bar f_h\}$. Thus the $S^{h,k}$ for all $l \in [K]$, together with $T^h$, form a partition $[\eta - h]$. 
The problem
    \begin{align*}
        \min\;\;  &  \sum_{i \in T^h} \bar f_i y_i  &\\
    \text{s.t.}\;\; & \sum_{i \in T^h} \bar w_{i} y_i  \ge \sum_{i \in T^h} \bar w_i \Bar{y_i}   &  \\
    &0 \le y_i \le 1 & \forall i \in T^h 
\end{align*}
is a linear relaxation of a knapsack cover problem, and hence it has an optimal solution $y^*$ with at most one fractional variable. We set $\Hat{y}_i^{h,\sigma} = \lceil y^*_i\rceil$ for each $i \in T^h$; thereby increasing cost (from $y^*$) by less than $(1 + \epsilon)^{-K} \bar f_h \le \epsilon \bar f_h$ by definition of $K$ and $\Hat{\alpha}^{h,\sigma} = \Bar{\alpha}$. 
In summary, 
\begin{align}
\bar f^T \Hat{y}^{h,\sigma} + \bar v \Hat{\alpha}^{h,\sigma} - \bar f^T \Bar{y} - \bar v \Bar{\alpha}& = \sum_{k: \sigma_k > 0} \left( \sum_{i \in S^{h,k}} \bar f_i \Hat{y_i}^{h,\sigma} - \sum_{i \in S^{h,k}} \bar f_i \Bar{y_i} \right) + \sum_{i \in T^h} \bar f_i (\Hat{y}_i^{h,\sigma} - \Bar{y}_i) \nonumber\\
   &\le \epsilon \sum_{k: \sigma_k > 0} \sum_{i \in S^{h,k}} \bar f_i \Bar{y_i} + \epsilon \bar f_h \label{44} \\
&\le \epsilon \bar f^T \Bar{y} \label{45}
\end{align}
Here, the equality follows from $\Hat{\alpha}^{h,\sigma} = \Bar{\alpha}$, \eqref{44} follows from Lemma~\ref{lemma:appendix1} and by definition of sets $S^{h,k}$ and~\eqref{45} follows from the fact that $\Bar{y}_h = 1$, by definition of $ P^{h,\sigma}$.
\end{proof}
Consider the polyhedron 
\[
Q = \conv \left( \bigcup_{h \in [\eta], \sigma \in ([J]\cup\{0\})^K} P^{h,\sigma} \right)
\]
Note that there are at most $(J + 1)^K \eta = \bigO\left(\epsilon^{(\log \epsilon/ \log (1+\epsilon)) - 1}\eta \right)$ polyhedra $P^{h,\sigma}$, and that each $P^{h,\sigma}$ is described by a system with $\bigO(K + \eta)$ constraints in $\eta$ variables. 
Thus, $Q$ is the projection to $\mathbb{R}^{\eta+1}$ of the feasible set for a system with at most $\bigO\left(\epsilon^{(\log \epsilon/ \log (1+\epsilon)) - 1}\eta^2 \right)$ constraints and variables. Furthermore, any 0/1 vector $y$ and $\alpha$ that is feasible for \mkcone~ satisfies $y \in P^{h,\sigma}$ for some $h$ and $\sigma$. In other words, $Q$ constitutes a valid relaxation to \mkcone. 
\begin{lemma}
    $\OPT^Z \le (1 + \epsilon) \min \{ \bar f^T y + \bar v \alpha : (y,\alpha) \in Q\}$.
\end{lemma}
\begin{proof}
    Let $(\dot{y},\dot{\alpha}) \in Q$ be an optimal extreme point. 
    It follows from the definition of~$Q$ that every extreme point of~$Q$ is also an extreme point of $P^{h,\sigma}$ for some $h \in [\eta]$ and signature~$\sigma$. 
    This implies that $(\dot{y},\dot{\alpha})$ is an optimal solution to $\min\{\bar f^T y + \bar v \alpha :  (y,\alpha)  \in P^{h,\sigma}\}$.
    By Lemma~\ref{lemma:appendix2}, there is a feasible solution $(\hat y, \hat \alpha)$ to \mkcone\ such that 
        \(\bar f^T \hat y + \bar v \hat \alpha \leq (1+\epsilon)\min\{\bar f^T y + \bar v \alpha :  (y,\alpha)  \in P^{h,\sigma}\} = (1+\epsilon)\min \{ \bar f^T y + \bar v \alpha : (y,\alpha) \in Q\} \).
\end{proof}

\end{document}